\documentclass[aps,pra,twocolumn,amsmath,showpacs,
amssymb,amsfonts,superscriptaddress]{revtex4}

\usepackage{graphics}

\usepackage{epsfig}
\usepackage{amsmath}
\usepackage{amsthm}
\usepackage{color}

\newcommand{\beq}{\begin{equation}}
\newcommand{\eeq}{\end{equation}}
\newcommand{\beqa}{\begin{eqnarray}}
	\newcommand{\eeqa}{\end{eqnarray}}

\newtheorem{theorem}{Theorem}

\newtheorem{defn}{Definition}

\begin{document}

\title{
Communication Through a Quantum Link
}

\date{\today}

 \author{Vittorio \surname{Giovannetti}}
    \email{v.giovannetti@sns.it}

\address{NEST CNR-INFM \&  Scuola Normale Superiore,
I-56126 Pisa, Italy}
 
 \author{Daniel \surname{Burgarth}}
\email{daniel.burgarth@maths.ox.ac.uk}
\address{Mathematical Institute, University of Oxford,
24-29 St Giles¡Ç Oxford OX1 3LB, UK}

\author{Stefano \surname{Mancini}}
\email{stefano.mancini@unicam.it}
\address{Dipartimento di Fisica, Universit\`a di Camerino,
I-62032 Camerino, Italy}

\begin{abstract}
{ A chain
of interacting spin behaves like a quantum mediator (quantum link)  which allows two distant parties that control the ends of the chain to exchange quantum messages. 
We show that over repeated uses without resetting the study of a quantum link  
can be connected to correlated quantum channels with
finite dimensional environment (finite memory quantum channel).
Then, using coding arguments for such kind of  channels
and results on mixing channels we  present a protocol that
allows us to achieve perfect information transmission through a quantum link. }
\end{abstract}


\maketitle

\section{Introduction}

{
Recently, an increasing attention has been devoted to interacting quantum systems in order to accomplish communication tasks. In fact, their evolution by means of quantum interference effects naturally leads to information transfer from one to another part located in a different place. A paradigmatic example is a chain of interacting spins (or more broadly speaking a spin network) -- see Ref.~\cite{BOSEREV} and references therein.
Here two distant parties (say the sender Alice and the receiver Bob) try to exchange quantum messages by operating on separate ends of  a chain of interacting qubits. 
Therefore the chain behaves like a mediator of quantum information, or like a \emph{quantum link}.

The information Alice sends through the link can get stuck into the link itself, thus resulting in imperfect transmission.  The faithfulness of information transfer has been widely investigated by  considering the link to be reset on each use 
either  by means of some external control operating directly on the whole chain, or by means of some
clever but costly "downloading" procedure~\cite{BOSEREV}.
A multi-use quantum communication scenario without resetting is intriguing as well. 
Actually a spin chain without resetting has been proposed as a physical model for quantum channel with memory \cite{BM04}. A preliminary study of such complex
communication lines has been carried out
in Ref.~\cite{RGF} by computing the transmission rates (i.e. number of transferred qubits per unit time) of some simple multi-use  protocols, and in  Ref.~\cite{BBMB08}
by focusing on the two channel uses scenario of some specific spin chain models.
Moving from such arguments, we study here the asymptotic (large number of uses)  behavior of a quantum link without resetting.

In particular we shall establish a connection between quantum link communication and 
a special class of correlated quantum channels, the finite memory channels,  that allows us to 
devise a new communication strategy. 
Indeed using coding arguments for finite memory
 channels and some results from mixing channels we  present a protocol that
allows us to achieve perfect information transmission through the quantum link. 
It results in the first (efficient) communication protocol for the multi-use scenario of spin chains.

The layout of the paper is the following.
In Section II we introduce the notion of perfect memory channels and we discuss coding arguments for them. In Section III we present a general communication scheme through a quantum link.
Then, in Section IV we perform an information flow analysis by using the coding arguments previously developed and results on mixing channels.
Finally, in Section V we present a protocol which allows us to achieve perfect information transmission through a spin chain. Section VI is for conclusions. 
}

\section{Perfect memory channels}
\label{sec:PM}

In the multi-use quantum communication scenario the sender of information  Alice  transmits  (classical or quantum) messages to her intended receiver Bob
by encoding them in the internal states of a (possibly infinitely long) ordered {\em sequences} of  information carriers
${X}:={x}_1, {x}_2,  \cdots$.  The latter are described as identical quantum systems characterized by the
 Hilbert spaces ${\cal H}_{{x}_1}$, ${\cal H}_{{x}_2}$,
 $\cdots$ of { the same} dimension $d:=\dim[ {\cal H}_{x}]$.
 Owing to the  noise that affects the communication, the messages received by Bob are a corrupted version of the input signals.
This process is formally described by assigning a {\em multi-use quantum channel}, i.e.  a collection 
${\cal L}:=\{ \Lambda^{(n)} : n\in \mathbb{N}\}$
 of completely positive trace-preserving (CPTP) maps~\cite{NIELSEN} $\Lambda^{(n)}$  connecting the input states of the carriers with their output 
 counterparts.
Specifically for any { positive integer} $n$, $\Lambda^{(n)}$ is the transformation 
that operates on the density matrices  $\rho_{X}$  of the Hilbert space 
${\cal H}_{X}^{(n)} = {\cal H}_{{x}_1} \otimes \cdots \otimes  {\cal H}_{{x}_n}$
associated to  the first $n$ carriers of ${X}$, i.e. 
\begin{eqnarray}
\Lambda^{(n)} : \qquad \rho_{X} \;\;\longrightarrow \;\;\Lambda^{(n)}(\rho_{X})\;, \label{mp1}
\end{eqnarray}
 under the minimal consistency requirement that
 the output $\Lambda^{(n-1)}(\rho_{X}^{(n-1)})$ associated with the
 first $n-1$ carriers should be obtained from~Eq.~(\ref{mp1})
by taking the partial trace with respect to  the $n$-th carrier.
In this context  ${\cal L}$  is said to represent  a {\em memoryless quantum channel} if for all $n$ the transformation~(\ref{mp1}) 
can be described  as a {\em tensor product} of the  CPTP map
$\Lambda:=\Lambda^{(1)}$ 
that acts on the states of the first carrier, i.e. 
\begin{eqnarray}
\Lambda^{(n)}(\rho_{X}) =\Lambda^{\otimes n} (\rho_{X})\;, \label{mp2}
\end{eqnarray}
with $\Lambda^{\otimes n} := \Lambda \otimes \cdots \otimes \Lambda$.
If Eq.~(\ref{mp2}) does not apply one says instead that ${\cal L}$
represents a {\em correlated} channel. Furthermore one says that it
  represents a {\em memory} channel if the sequence
$\Lambda^{(1)}, \Lambda^{(2)}, \cdots,$ possesses a causal
structure (i.e. if for all $n$, the output states of the first $n$-th carries
${x}_1, {x}_2, \cdots, {x}_n$ do not depends upon the input states of the subsequent carriers)~\cite{BM04,KW05,vg05}.

It is well known~\cite{LIND,Stine} that any CPTP map admits  unitary dilations
that allow one to represent it in terms of  a unitary coupling with an 
external environment. In particular 
for the $n$-th element of ${\cal L}$
 we can write
 \begin{eqnarray}
\Lambda^{(n)}(\rho_{X}) &=& \mathrm{Tr}_{Y} \Big[ 
U_{{X}Y}
 \big(\rho_{X} \otimes 
\omega_Y \big)
U_{{X}Y}^\dag \Big]\;,
\label{eqn:memory_model}
\end{eqnarray}
where  $U_{{X}Y}$  is a
unitary that couples the $n$ carriers' { state $\rho_{X}$} to the 
 state  $\omega_Y$ of a multi-use environment $Y$ described by the
 Hilbert space  ${\cal H}_Y^{(n)}$.
 Upon purification  one can always choose $\omega_Y$ to be a pure vector $|\omega\rangle_Y$: when this happens the dilation~(\ref{eqn:memory_model}) is said to be of Stinespring form~\cite{Stine} and it is unique up to some irrelevant isometry acting on the environment $Y$.
 The unitary dilations~(\ref{eqn:memory_model}) can also be  put in a one-to-one correspondence
with the {\em operator sum} (or {\em Kraus}) representations of $\Lambda^{(n)}$~\cite{KRAUSBOOK},
  \begin{eqnarray}
\Lambda^{(n)}(\rho_{X}) &=& \sum_{{j}=0}^{d_Y^{(n)}-1}
\; \; K_{j} \; \rho_{X} \; K_{j}^\dag \label{eqn:kr_memory_model}\;,
\end{eqnarray}
where $d_Y^{(n)} :=\mbox{dim}[{\cal H}_Y^{(n)}]$ and
$K_{j} := {_Y\langle}\xi_{j}| U_{{X}Y} |\omega \rangle_Y$ with
 $\{ |\xi_{j}\rangle_Y; j=0, \cdots, d_Y^{(n)}-1\}$ being an orthonornomal basis of ${\cal H}_Y^{(n)}$.

From the uniqueness of Stinespring representation~\cite{Stine} it
follows that 
(apart from the trivial case
  of noiseless, or unitary, transformations), the memoryless channels are
  characterized by possessing unitary dilations in which the environment has a dimension which is exponentially growing in $n$ (i.e.
  $\log_2[ \mbox{dim}{\cal H}_Y^{(n)} ] = n \log_2[ \mbox{dim}{\cal H}_Y^{(1)}]$)
  or, equivalently, by possessing a (minimal) operator sum representations whose Kraus sets contains a number of elements which
  is  exponentially growing in $n$.
  This same property typically holds also for memory  channels with the important exception of the so called {\em perfect memory channels}~\cite{BM04,KW05,vg05}. They are characterized by the property of admitting
  unitary dilations~(\ref{eqn:memory_model}) in which the dimension of the environment ${\cal H}_Y^{(n)}$ is constant in $n$ (the extremal case
  being represented by the noiseless channels in which $\mbox{dim}[{\cal H}_Y^{(n)}] =0$ for all $n$).
 Such class of channels may be extended to include all sequences of CPTP maps that have a representation with a finite upper bound on the dimension of the environment state. 
More specifically:
\begin{defn}\label{def:memch}
A multi-use quantum communication channel ${\cal L}$  defined by the sequence of CPTP maps $\Lambda^{(n)}$ is termed as a {\em Perfect Memory (PM)} channel if there exists a sequence of unitary representations~(\ref{eqn:memory_model})
of the $\Lambda^{(n)}$s
such that,
\begin{equation}\label{pro1}
\lim_{n\rightarrow \infty} \frac{1}{n}\log_2 [d_Y^{(n)}] = 0 \;,
\end{equation}
where $d_Y^{(n)}$ are the dimension of the  environmental Hilbert spaces  $\mathcal{H}_Y^{(n)}$ 
that enter in the dilation. Equivalently ${\cal L}$ 
is said {\em PM}  if it admits a sequence of operator sum representations characterized by a number of elements $d_Y^{(n)}$ which satisfy Eq.~(\ref{pro1}).
\end{defn}
Physically speaking the property~(\ref{pro1})  means that the size of the environment $E$ does not grow fast enough to capture all the information that is sent through the channel (the latter being measured by the size of 
the carriers ${\cal H}_{X}^{(n)}$ which is exponential in $n$). Intuitively we thus expect that  PM channels should allow for
efficient communication between Alice and Bob. This was formalized in Ref.~\cite{KW05} by showing that PM channels
have indeed optimal (classical and quantum) transmission rates, i.e. allow the transfer of $\log_2 d$ qubit  per channel use in the asymptotic limit of large $n$. 
 
In particular one can verify that the following theorems hold:

\begin{theorem}
Let ${\cal L}$ be a multi-use PM channel
operating on information carriers of dimension $d$  and let  $\{ d_Y^{(n)}: n\in \mathbb{N}\}$ be  the sequence satisfying Eq.~(\ref{pro1}).
Then for sufficiently large $n$
there exists a zero-error classical code ${\cal C}^{(n)}_X$ of size
\begin{equation}
|{\cal C}^{(n)}_X| \geqslant \frac{d^n}{(d_Y^{(n)})^2}\;,
\end{equation}
corresponding to a bit  transmission   rate
$R_{\mathcal{C}}^{(n)} := \tfrac{1}{n} \log_2 |{\cal C}^{(n)}_X|$ 
that  converges to the
optimal value $\log_2 d$  for $n\rightarrow \infty$.
\end{theorem}
\begin{proof}
Given { a positive integer} $n$, let $K_j$ the $d_Y^{(n)}$ 
 Kraus operators associated with $n$-th element of ${\cal L}$, i.e.  the CPTP map $\Lambda^{(n)}$ that operates on the first $n$ carriers.
A zero-error classical code that corrects the noise introduced by $\Lambda^{(n)}$
is a collection ${\cal C}^{(n)}_X$ of (orthonormal) codewords $|c_k\rangle_X \in {\cal H}_{X}^{(n)}$ which must obey the conditions,
\begin{equation}
{_X\langle} c_k| K_i^{\dag}K_{j}|{c_{k'}}\rangle_{X} = \delta_{kk'}\; M_{ij}{(k)}\;,
\label{eqn:classcond}
\end{equation}
for all codewords $|c_k\rangle_X$, $|c_{k'}\rangle_X$, and for all $i,j$, 
with $M_{ij}{(1)}, M_{ij}(2), \cdots,$ being $(d_Y^{(n)})^2 \times  (d_Y^{(n)})^2$ Hermitian matrices~\cite{NIELSEN,KLV00}.  These conditions imply that the support of the output states for all codewords are orthogonal, and hence may be distinguished with zero probability of error.
Suppose to have found $\ell$ orthogonal codewords $|c_1\rangle_X,\ldots,|c_\ell\rangle_X$ satisfying the condition \eqref{eqn:classcond}
(this is always possible for at least $\ell=2$). Then an additional codeword $|c_{\ell+1}^{(n)}\rangle_X$ can be chosen such that 
\begin{equation}
{_X \langle} c_{\ell+1}| K_i^{\dag}K_{j} |c_k\rangle_X =0 \;,
\label{eqn:perpcond}
\end{equation}
for all $i,j$ and for all $k=1,\cdots, \ell$.
Such a state exists provided that the total number of vectors $K_i^{\dag}K_{j}|c_k\rangle_X$ is less than or equal to
 the dimension $d^n$ of the input space ${\cal H}_{X}^{(n)}$, i.e. 
$\ell (d_Y^{(n)})^2\leqslant d^n$ (notice that, due to the sub-exponential character of $d_Y^{(n)}$, for PM channels this inequality can be always satisfied for some positive $\ell$, if  $n$ is sufficiently large). Hence we may continue the procedure until the set of codewords cannot be extended. That is we can get a code with at least  $d^n/(d_Y^{(n)})^2$ orthogonal codewords that can be transmitted with zero error. This corresponds to a rate (i.e. ratio of the faithfully transferred classical bits over the number of channel uses) larger than $\log_2 d - \frac{2}{n}\log_2 d_Y^{(n)}$ which converges to the maximum value $\log_2 d$ attainable when using $d$-dimensional information carriers.
\end{proof}

To construct a quantum error correcting code, we can utilize the result of Ref.~\cite{KLV00} in building a quantum code from the existing zero-error classical code.

\begin{theorem}
Let ${\cal L}$ be a multi-use PM channel 
operating on information carriers of dimension $d$  and
characterized by the sequence $\{ d_Y^{(n)}: n\in \mathbb{N}\}$  { satisfying} Eq.~(\ref{pro1}).
Then for sufficiently large $n$
there exists a zero-error quantum error correcting code ${\cal Q}^{(n)}_X$ of size
\begin{equation}
|{\cal Q}^{(n)}_X| \geqslant \frac{d^n}{(d_Y^{(n)})^4+(d_Y^{(n)})^2}\;,
\end{equation}
corresponding  to a qubit transmission rate 
  $R_{\mathcal{Q}}^{(n)} := \tfrac{1}{n} \log_2 |{\cal Q}^{(n)}_X|$ 
that converges to the
optimal value $\log_2 d$  for $n\rightarrow \infty$.
\end{theorem}
\begin{proof}
Under the same definitions of Theorem 1, the conditions for an error correcting {\em quantum} code  are given by~\cite{NIELSEN,KLV00}
\begin{equation}
{_X\langle} q_k| K_i^{\dag}K_{j}|{q_{k'}}\rangle_X = \delta_{kk'}\; M_{ij}\;,
\label{eqn:quacon}
\end{equation}
where now the quantum codewords $|q_k\rangle_X$ are a basis for the coding subset ${\cal Q}^{(n)}_X\subseteq {\cal H}_{X}^{(n)}$ and with  $M_{ij}$ a  matrix that {\em does not} depend on the index $k$ of the quantum codewords 
$|q_k\rangle_X$.   
Consider then the classical code ${\cal C}^{(n)}_X$ that we have constructed in the derivation of Theorem 1.
Let us divide it in $\ell$ non overlapping subsets ${\cal C}_X^{(n)}(1)$, $\cdots$, ${\cal C}_X^{(n)}(\ell)$. For each $k=1,\cdots, \ell$
define also the vector $|q_k\rangle_X$  to be a (proper) superposition of the classical codewords that belong to ${\cal C}_X^{(n)}(k)$.
One can easily verify that the set $\{ |q_1\rangle_X, \cdots, |q_\ell\rangle_X\}$ still satisfies the classical code conditions~(\ref{eqn:classcond})
with  matrices $M_{ij}{(k)}$ which are convex convolutions of the previous ones.
The idea is thus to select the partitions ${\cal C}_X^{(n)}(1)$, $\cdots$, ${\cal C}_X^{(n)}(\ell)$ and the associated vectors $\{ |q_1\rangle_X, \cdots, |q_\ell\rangle_X\}$  in such a way that the new matrices $M{(k)}$ will be all identical: if this happens
the vectors $\{ |q_1\rangle_X, \cdots, |q_\ell\rangle_X\}$ will automatically satisfy the condition~(\ref{eqn:quacon}).
 In the end the problem can thus be mapped into a convex
optimization problem:   invoking  a theorem by the Radon~\cite{RADON} it is then possible to show that it admits a solution provided
that  $\ell [(d_Y^{(n)})^4+(d_Y^{(n)})^2]< d^n$ (as before this inequality makes sense if ${\cal L}$ is PM at least 
for large enough $n$).
In this way we can get a code with $d^n/[(d_Y^{(n)})^4+(d_Y^{(n)})^2]$ orthogonal codewords that can be transmitted with zero error. Hence the rate $\log_2 d - \frac{1}{n} \log_2 [(d_Y^{(n)})^4+(d_Y^{(n)})^2]$ can be attained.
If the final term scales such that it vanishes in the asymptotic limit, then the channel is asymptotically noiseless.  This applies to all cases where $d_Y^{(n)}$ is sub-exponential in $n$.
\end{proof}

\subsection{Decoding transformation for PM channels}\label{sec:dec}

Even though this is a straightforward application of quantum error correction procedures~\cite{NIELSEN} it is  useful to give a close look at decoding strategy associated with Theorem 2. 
By construction the $(d_Y^{(n)})^2 \times  (d_Y^{(n)})^2$ matrix  $M_{ij}$ of Eq.~(\ref{eqn:quacon})
 is positive semidefinite. Take  then $O_{ij}$ the unitary matrix that diagonalizes $M_{ij}$ 
and define the $d_Y^{(n)}$ operators 
$F_{j} := \sum_{i}
O_{ij} \; K_{i}$. They provide an alternative sum operator 
decomposition of $\Lambda^{(n)}$
 and satisfy the orthogonality condition
\begin{equation}
{_{X}\langle} q_k|F_i^{\dag}F_{j}|{ q_{k'}}\rangle_{X} = \delta_{kk'}\; \delta_{ij} \; \lambda_j \;,
\label{eqn:quacon11}
\end{equation}
for all the vectors $|q_k\rangle_X$ that form a basis of ${\cal Q}^{(n)}_X$, with $\lambda_j\geqslant 0$ the eigenvalues of $M_{ij}$.
By  polar decomposition we can then write 
\begin{eqnarray}\label{EEE2}
F_j P = \sqrt{\lambda_j} \; U_j \; P = \sqrt{\lambda_j} \; P_j \; U_j\;,
\end{eqnarray}
with $U_j$ being a unitary transformation, $P$ being the projector on ${\cal Q}^{(n)}_X$ { such that} $P_j :=U_j P U_j^\dag$. From Eq.~(\ref{eqn:quacon11}) it follows that the projectors $P_j$ are orthogonal,
i.e. $P_i P_{j} = \delta_{ij} \; P_j$. Consider thus a generic state $\rho_{X}$ of ${\cal Q}^{(n)}_X$, i.e. 
$P \rho_{X} P = \rho_{X}$. Equation (\ref{EEE2}) allow us the to express the output state 
as 
\begin{eqnarray}\label{eee}
\Lambda^{(n)} (\rho_{X}) = \sum_j \; \lambda_j \; P_j\; \big[ U_j \; \rho_{X}\;  U_j^\dag \big] \; P_j^\dag \;,
\end{eqnarray}
which is explicitly written in block form thanks to the orthogonality conditions of the $P_j$s.
We can hence recover $\rho_{X}$ through the following steps:
first perform a projective measurement on $\Lambda^{(n)} (\rho_{X})$
that distinguishes among the
orthogonal subspaces of ${\cal H}_{X}^{(n)}$ associated with the $P_j$.
With probability $\lambda_j$ we will get the outcome $j$. Apply then the unitary rotation $U_j^\dag$ to the projected state:
independently from the measurement outcome the final state will  be transformed in the input message $\rho_{X}$.

Before concluding this section we would like to spend few words on the global evolution that $X$ and $Y$ undergo during the
encoding-decoding stages: this will play an important role in the subsequent sections.
Consider then $|\psi\rangle_{X}\in{\cal Q}^{(n)}_X$. Using the unitary representation~(\ref{eqn:memory_model}) associated with the Kraus set $K_j$ and exploiting the above identities 
we get:
\begin{eqnarray}
&&U_{{X}E} (|\psi\rangle_{X} \otimes | \omega\rangle_Y)
= \sum_{j}  (K_j | \psi\rangle_{X})  \otimes |\xi_j\rangle_Y \nonumber
\\ &&\qquad = \sum_J \sqrt{\lambda_j} 
\; (P_j \; U_j  | \psi\rangle_{X})  \otimes | \zeta_j\rangle_Y \;,\label{eewe}
\end{eqnarray}
where $|\zeta_j\rangle_Y:= \sum_{i} O_{ij}^* |\xi_i\rangle_Y$ is an orthonormal set of ${\cal H}_Y^{(n)}$.
Of course by taking the partial trace with respect to $Y$ yields the final state of Eq.~(\ref{eee}). 
What it is interesting for us however is to observe that {\em before} the decoding stage $X$ and $Y$ 
are in general entangled. Furthermore we notice that due to the orthogonality condition of the $P_j$ and $|\zeta_j\rangle_Y$,
  the vector~(\ref{eewe})  is automatically written in { the} Schmidt form.

\section{Communication by a quantum mediator}\label{cbqm}
As anticipated  in the introduction quantum networks communication~\cite{BOSEREV}
can   be seen as a particular instance of the communication scenario sketched in Fig.~\ref{fig2}.
Here the {\em mediator} $M$ is a composite 
quantum object of {\em finite} dimension $d_M$, that is composed by 
 three subsystems $M_A$, $M_C$, and $M_B$ which interact through some 
given Hamiltonian  $H$.  $M$ acts as  an effective quantum channel that connects two distant parties, the sender of information Alice and the receiver Bob
who are supplied with the quantum registers $A$ and $B$, respectively.
The register $A$ is assumed to be composed by a sequence of ordered memories $a_1,  a_2, \cdots$. In the communication scenario we consider, Alice "writes" on $A$ the quantum messages
she wants to communicate to Bob. 

\begin{figure}[t!]
\begin{center}
\includegraphics[scale=.7]{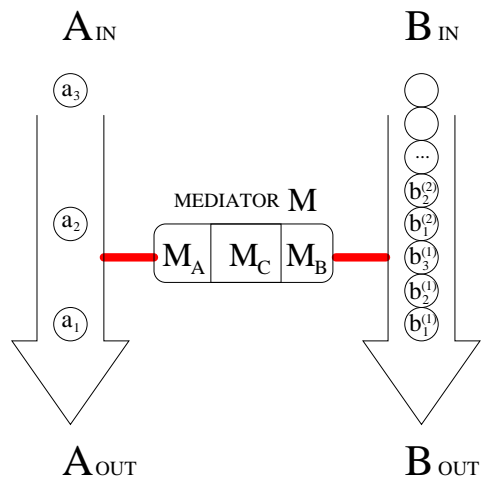}
\caption{Communication through a quantum link. Alice sends to Bob the messages she has stored in the quantum memories $a_1,a_2, \cdots a_n$
by coupling them  with the subsystem $M_A$ of the mediator $M$; each memory { element} interacts only once with $M_A$ following the sequential order
indicated by the left arrow of the figure 
(i.e. first $a_1$, then $a_2$, and so on). Bob recovers  Alice's messages by preparing each one of his quantum memories $b$ in the same fixed fiduciary state $|\nu\rangle$ and by 
coupling them with the subsystem $M_B$ of
$M$ (again the coupling will follow the sequential order indicated by the right arrow). The red lines of the figure represent the couplings between the registers
and $M$. $A_{in}$, $B_{in}$ represent the input ports of the device which are used by Alice and Bob to bring  their qubits in contact with $M$. Similarly 
$A_{out}$ and $B_{out}$ represent the output ports from which  the qubits emerge after their interaction with the mediator.  
} \label{fig2}
\end{center}
\end{figure}
The memories $a_1, a_2, \cdots$  are then sequentially put in contact with
 the subsystem  
$M_A$ of 
$M$ following a fixed schedule in which first  $M_A$ couples with  $a_1$, then with $a_2$, etc. Such interactions are assumed to be identical ({\em uniform coupling regime}), and faster and stronger than the free  evolution of the mediator. Consequently we will 
represent them in terms of a collection of two-body  gates $S_{a_1}, S_{a_2}, \ldots, S_{a_n}$ that connect the corresponding memory elements and $M_A$. 
As indicated  in Fig.~\ref{fig2} the whole process
can  be described  as if the memories were moving  in a line from the port $A_{in}$ to the port $A_{out}$ and interact with $M_A$ only when passing close to it.
Conversely  Bob's register $B$ is composed by a sequence of memories each initialized in the same fiduciary state $\nu:= |\nu\rangle\langle \nu|$.
They are grouped in independent and not necessarily uniform sub-registers $B^{(1)}, B^{(2)}, \cdots$. The $k$-th sub-register $B^{(k)}$ contains
$m_k$ memories indicated by the subscript $b_1^{(k)}, b_2^{(k)}, \cdots, b_{m_k}^{(k)}$ which are supposed to receive and store the info contained in Alice's $k$-th input qubit $a_k$~\cite{NOTA2}. To do so, during the time interval that elapses between Alice $k$-th operation and the subsequent one, 
Bob will put the  qubits  of the sub-register $B^{(k)}$ in contact with $M$ by applying 
a series of fast quantum gates  $S_{b^{(k)}_\ell}$ that couple the memories $b^{(k)}_\ell$ with the subsystem $M_B$
(again one can represent this process as if Bob's memory elements were propagating from the port $B_{in}$ of Fig.~\ref{fig2}
to the port $B_{out}$ interacting with $M_B$ only when passing close to it).
For each $k\in \mathbb{N}$ this yields the following unitary transformation acting on $M$ and $B^{(k)}$ 
\beq\label{vk}
V_{k}:=
S_{b^{(k)}_{m_k}} e^{- i H \tau}
\ldots 
S_{b^{(k)}_2}e^{- i H \tau}
 S_{b^{(k)}_1} e^{- i H \tau}
 \;,
\eeq
where $e^{ - i H \tau}$ describes the free evolution of $M$  between two consecutive quantum gates.
In writing Eq.~(\ref{vk})  we have assumed that, as in Alice case, the couplings introduced by Bob are uniform and operate on a time scale much shorter than those of $H$. We assumed also that 
uniform time intervals $\tau$ elapses among any two consecutive $S_{b^{(k)}_{\ell}}$. The global unitary transformation of   $ABM$ 
during the transmission  can  then be expressed by composing the transformation $V_k$ with Alice's gates
$S_{a_k}$.  Specifically, 
suppose that Alice  use only  the first  $n$ elements of  $A$ for the communication and suppose that she
prepares them in the state $\rho_A$.
The initial state of the global system reads then 
$\rho_A\otimes \nu^{\otimes m}_B
\otimes \omega_M$,
with  $m= \sum_{k=1}^n m_k$ being the total number of Bob's memories that play an active role  in the
transmission of the first $n$ elements of $A$, $\nu^{\otimes m}_B$ being their input state 
 $\nu^{\otimes m}_B:=\nu_{b_1^{(1)}} \otimes \cdots  \otimes \nu_{b_{m_n}^{(n)}}$, 
 and with $\omega_M:= |\omega\rangle_M \langle \omega|$  being the initial state of the mediator $M$ that  we will
assume to be pure (the generalization to the mixed case being trivial).
Bob's output states are thus given by 
\beq \label{sigmaB}
\Lambda_{A\rightarrow B}^{(n)}(\rho_A)
=\mbox{Tr}_{A, M}\left[ W (
 \rho_A\otimes  \nu^{\otimes m}_B
\otimes \omega_M
 ) W^\dag 
\right] \;,
\eeq
with $\mbox{Tr}_{A,M}$ being the partial trace with respect to $A$ and $M$ and with
\begin{eqnarray}
W= 
V_{n}S_{a_n}
\ldots V_{2}S_{a_2}V_{1} S_{a_1}\;.
\end{eqnarray}

\section{Information flow analysis}\label{sec:flow}

In this section we present a detailed analysis  of the communication model of Fig.~\ref{fig2} that allows us
to identify the two main  mechanisms that superintend at the information flow through such a scheme.
In particular using the results on PM channels presented in Sec.~\ref{sec:PM} we will show that there are efficient, finite size encodings, that
allows one to  prevent Alice messages from getting stuck into the mediator $M$. By itself this does not ensure
perfect transfer from $A$ to $B$. However, using ideas from the theory
of mixing channels~\cite{RAG, GOHM,NJP,TDV} we will show that one can force Alice messages to  
focus on Bob's memories.  

Equation~(\ref{sigmaB}) defines a quantum channel  $\Lambda_{A\rightarrow B}^{(n)}$ 
which connects the input port $A_{in}$ of Fig.~\ref{fig2} to the output port $B_{out}$, i.e. 
which
takes the input state of the $n$ memories of the register $A$ to the output state of the $m$ memories of  $B$. 
It is clearly a memory channel with $M$  and $A$ playing the role of
the multi-use environment $Y$ and with the memory effects arising from the  possibility that part of the signals encoded in some earlier $a_k$  will get stuck 
in $M$ interfering with the subsequent ones.
  As mentioned in the introduction, by identifying $M$ with a  network of interacting qubits, most of the spin  communication protocols~\cite{BOSEREV} can be represented in this model (an explicit example is presented
  in the next section).
  Having fixed the Hamiltonian $H$ an interesting problem  is then to determine wether or not  there exist suitable choices of for the local transformations $S_{a_k}$ and $S_{b_\ell^{(k)}}$,  the timing $\tau$, the encoding $\rho_A$ and possibly
the fiduciary state $\nu$ of the $B$ memories, that allows for a  {\em reliable} and possibly {\em efficient} information transmission from $A_{in}$ to $B_{out}$
(reliability referring to the possibility of achieving perfect transmission fidelity, efficiency referring instead to the effective number of
memory elements -- or coupling operations $S_a$ --  per transmitted qubit
Alice needs to use). Even though one can easily find examples of $H$ which admits simple answers for the above questions, in the general case this is  not an easy problem to solve. 
Interestingly enough however  that same questions admit an elegant solution if one consider the transmission
efficiency of the channel  $\Lambda_{A\rightarrow AB}^{(n)}$
which connects the input port $A_{in}$ to the joint output ports $A_{out}+B_{out}$ of Fig.~\ref{fig2}. In this case only $M$ plays the role of the environment $Y$ yielding the 
transformation 
 \beq \label{sigmaAB}
\Lambda_{A\rightarrow AB}^{(n)}(\rho_A)
=\mbox{Tr}_{M}\left[ W (
 \rho_A\otimes
  \nu^{\otimes m}_B
\otimes\omega_M) W^\dag 
\right] \;.
\eeq
Since $M$ is a finite dimensional system (say of dimension $d_M$) the map $\Lambda_{A\rightarrow AB}^{(n)}$ 
is a Perfect Memory channel characterized by a Kraus set composed by  a number of
elements (i.e. $d_M$) which is constant in $n$.
According to { Section \ref{sec:PM}}, we know that there exists a zero-error quantum error correcting code { ${\cal Q}_A^{(n)} \in {\cal H}_A^{(n)}$} of size
{ $|{\cal Q}_A^{(n)}| \geqslant \tfrac{d^n}{d_M^2(d_M^2 +1)}$} where $d^n$ is the dimension of the register $A$ (each memory having dimension $d$). 
Using such code Alice can reliably transfer info from $A_{in}$ to $A_{out}+B_{out}$ at a rate which is optimal 
for sufficiently large $n$ (i.e. it allows one to transfer one qubit per memory element). In particular consider the case of $\omega_M$ pure (the 
mixed case  can be treated analogously upon purification) and indicate with $P$ the projector { onto ${\cal Q}_A^{(n)}$}.
According to the analysis of Sec.~\ref{sec:dec} the global final state associated to a generic  input state $|\psi\rangle_A$ of the coding subspace ${\cal Q}_A^{(n)}$
admits then the following Schmidt decomposition
\begin{eqnarray}\label{imp}
&&W ( |\psi \rangle_{A} \otimes
|\nu \rangle_B^{\otimes m}
\otimes |\omega\rangle_M)  \\ \nonumber  
&&=\sum_{j} 
 \sqrt{ \lambda_j}   ( P_j  U_j \; |\psi \rangle_A \otimes
 |\nu\rangle_B^{\otimes m}
  ) \otimes |\zeta_j \rangle_M\;,
\end{eqnarray}
where $\sqrt{\lambda_j}$ are Schmidt coefficients, $U_j$ are  unitary transformations acting on $B$ and on the coding subspace of $A$, $P_j := U_j P U_j^\dag$ are orthogonal projectors 
on $AB$, and  $\{ |\zeta_j\rangle_M; j=0, \cdots, d_M-1\}$ is an orthonormal basis  of $M$~\cite{NIELSEN}.
Therefore for all density matrices $\rho_A$ of ${\cal Q}_A^{(n)}$
Eq.~(\ref{sigmaAB}) can be expressed in the following block form
\begin{eqnarray} \label{dd}
\Lambda_{A\rightarrow AB}^{(n)}(\rho_A)
= \sum_j
\lambda_j
\; P_j  U_j \; ( \rho_A \otimes
\nu_B^{\otimes m}
 ) \; U_j^\dag P_j  \;,
\end{eqnarray}
while the recovering of the information can be obtained by performing a (possibly joint) projective measurement
on $AB$ which distinguishes the orthogonal subspaces associated with the $P_j$. 
 
Of course from the prospective of using the link $M$ to transmit signals from $A_{in}$ to $B_{out}$  
the above result is quite unsatisfactory: {\em i)} it maps the messages into a possible joint subspace of $A$ and $B$, {\em ii)} it requires to perform joint operation on $A$ and $B$ to recover them.
 What is relevant for us, however, is the fact that using such encoding we can at least guarantee that Alice messages
 do not get stuck in  $M$.
Furthermore the size of the quantum error correcting code  { ${\cal Q}_{A}$}
is independent from the number $m$ of Bob's qubits, as long as he prepares such qubits into a fixed reference state  and couples them with $M_B$ through a 
sequence  of gates $S_{b_\ell^{(k)}}$ that are {\em known} to Alice. The real question is then determining whether or not   one can  force the output information  to be localized only on $B_{out}$. 

One indication that this may be possible  for at least some  spin network communication models considered so  far~\cite{BOSEREV}
comes from a change of perspective.
Consider in fact the channel that acts upon the mediator $M$ between two consecutive interactions with Alice's
qubits --- say the $k$-th and $(k+1)$-th interactions. This is the map that given a state $\omega'_M$ of $M$  takes it to  
\begin{eqnarray}
 \mbox{Tr}_{B^{(k)}} [ V_k (\omega'_M\otimes
 \nu_{B^{(k)}}^{\otimes m_k} 
  )V_k^\dag] 
= \underbrace{ {\cal N} \circ {\cal N} \circ \cdots \circ{\cal N}}_{m_k} ( \omega'_M) \;,
\nonumber \end{eqnarray}
where $\nu_{B^{(k)}}^{\otimes m_k} := \nu_{b_1^{(k)}}\otimes \cdots \otimes \nu_{b_{m_k}^{(k)}}$  is the input state of
the $k$-th sub-register and
"$\circ$" represents  superoperator composition and where ${\cal N}$ is the CPTP map
\begin{eqnarray}
{\cal N} (\omega_M') =\mbox{Tr}_b[ (S_b e^{-iH\tau})(\omega_M'\otimes  \nu_b
) (S_b e^{-iH\tau})^\dag]\;.
\end{eqnarray}
The above equations show that, under the repetitive interactions with the $B$ memories, the evolution of the mediator 
 $M$ can be described as a sequence of iterated  application of a channel. By general properties of CPTP maps it is know that 
in the asymptotic limit of $m_k\rightarrow \infty$  this will {\em typically} induce a relaxing behavior on $M$ which
will bring such system to a   final point  $\omega_M^*$ (the {\em fixed point} of the map) that is independent from the input state $\omega'_M$~\cite{RAG,GOHM,NJP,TDV}. 
This is known as {\em mixing} or {\em  relaxing} property of ${\cal N}$ and the associated convergency speed is know to the exponential fast in $m_k$ (typically referring to the fact that the vast majority of CPTP maps possess the mixing property).
We can hence use this result to say that for most choices of $S_b$ and $|\nu \rangle_b$, by choosing $m_k\gg1$ after any stage of Bob's coupling the system  $M$ can be brought {\em close}
 to a fix point $\omega_M^*$. In turn this implies that the channel $\Lambda_{A\rightarrow A}^{(n)}$ that connects  $A_{in}$ to $A_{out}$ will be {\em close}
 to the {\em memoryless} map  $\tilde\Lambda^{\otimes n}_{A\rightarrow A}$ with 
 $\tilde\Lambda_{A\rightarrow A} (\rho_A):= \mbox{Tr}_M [ S_A (\rho_A  \otimes \omega_M^*)S_A^\dag]$. 
For our purposes however the most appealing property of mixing channels is another one.
Suppose  in fact that the fix point  $\omega_M^*$ of ${\cal N}$ is  pure. In this case  one can verify that, in the asymptotic limit $m_k\rightarrow \infty$ all information contained in $M$ is transferred  to $B$ (the convergency speed being again exponentially fast in $m_k$). 
If the couplings $S_A$ that connects  $A$ with $M_A$ are then able to transfer Alice messages into $M$ we can use the mixing property of ${\cal N}$
to {\em drive} such information into $B$.

\section{Spin chain as a quantum link}

In the following we introduce a spin chain implementation of the quantum link model of Fig.~\ref{fig2} which allows one to
take full advantage of the analysis presented in the previous section. { As a result will be able to show that one can use such spin chain to reliably transfer quantum information from $A_{in}$ to $B_{out}$.}

The setup we are interested is similar to the one discussed in Ref.~\cite{DGB}: it uses as mediator $M$ a collection of $L$ identical, independent, $N$-long chains of $\frac{1}{2}$-spins. 
For $L=2$ it corresponds to the the Dual Rail spin chain communication line of Ref.~\cite{DB}: 
since the main properties of the model are already captured by the simpler case 
for the sake of simplicity in the following we will focus only on it
(the case  $L=1$~\cite{BOSE} has similar properties but it introduces some unnecessary complications in the discussion since
 it is lacking  of a fundamental symmetry in the messages encoding). The spins of the two chains that form $M$ are 
coupled through a (not necessarily first-neighbor)  Heisenberg-like Hamiltonian $H$ and are initially prepared into the {\em all-spin-down} configuration 
\begin{eqnarray}
|\omega\rangle_M 
:=\tiny{\left| \begin{array}{c}\downarrow \downarrow \cdots \downarrow \\ \downarrow \downarrow \cdots \downarrow \end{array}\right\rangle}\label{iniziale}\;,
\end{eqnarray}
which one can assume to be the ground state of $H$ (in writing the left hand side of this expression 
we have adopted the simplified notation
to represent the state of the first of the two spin chain with the elements of the first line, 
and the state of the second spin chain with the elements of the second line).
%
In this context the subsets $M_A$ and $M_B$ correspond, respectively, to the collection of the first spins and the last spins of the $L$ chains. To simplify the analysis we also assume Alice's and Bob's  memories to be qutrits characterized by the ortho-normal states $|0\rangle$, $|1\rangle$ and $|E\rangle$: Alice encodes her input messages into the subspace $\{ |0\rangle, |1\rangle\}^{\otimes n}$ of $A$ generated by the components $|0\rangle$ and $|1\rangle$ (therefore each input memory of $A$ will encode at most one qubit of info). Conversely we will identify $|E\rangle$ with the fiduciary state $|\nu\rangle$ of Bob's memories.
For the couplings $S_{a_k}$ and $S_{b_\ell^{(k)}}$ we chose
mappings which act as the identity  but for the following SWAP-like transformations, 
\beqa
|0\rangle \otimes \tiny{ \left|  \begin{array}{c}\downarrow\\ \downarrow\end{array}\right\rangle} \quad &\longleftrightarrow& \quad |E\rangle
\otimes 
\tiny{ \left| \begin{array}{c}\uparrow\\ \downarrow\end{array}\right\rangle} \;,\nonumber \\
|1\rangle \otimes \tiny{ \left| \begin{array}{c}\downarrow\\ \downarrow\end{array}\right\rangle} \quad 
&\longleftrightarrow&\quad  |E\rangle  \otimes 
\tiny{ \left| \begin{array}{c}\downarrow\\ \uparrow\end{array}\right\rangle} \;, \label{swap}
\eeqa
where we used the same  notation of (\ref{iniziale}) to represent the state of the spin chains.
Notice that such couplings allow the transferring of information from $a_k$ to $M_A$ only if  both spins of $M_A$ are pointing down:
furthermore when this happens $a_k$  moves from the message subspace spanned by the vectors $|0\rangle$ and $|1\rangle$ to the state $|E\rangle$. 
If instead $M_A$ has one or two spins up then the transmission is prevented and $a_k$ and $M_A$  keep their initial configurations. 
Similarly on Bobs side there is a net  flow of information form $M_B$ to $b_{\ell}^{(k)}$ only if the former has one spin up and one spin down: when this happens
  $b_{\ell}^{(k)}$ moves from the fiduciary state $|E\rangle$ to the message subspace and $M_B$ will be promoted to the all spin down state.
  In all the other cases instead $b_{\ell}^{(k)}$  and $M_B$ will keep their initial configurations.
It is worth commenting that this is a main difference with respect to other swapping strategies introduced in the past (e.g. see Ref.~\cite{NOSTRO}). In the present case
in fact, there is no guarantee that the repetitive application of the gates $S_{b_\ell^{(k)}}$ will drive the chain toward the ground state (i.e. the associated map
on $M$ is not necessarily mixing). Such property however still holds for at least the first excitation sector of the chains (this is the subspace of $M$ in which
both chains have exactly one spin  up each): as will be clear in the following this is enough to set up a reliable transmission protocol.

\subsection{Excitation distribution and Schmidt decomposition of the final state} 
We notice that  both the couplings~(\ref{swap})  and the spin chains Hamiltonian
 preserve a  global observable $Z$ which counts the number of "excitations" present in the system, i.e. 
\begin{eqnarray}
Z = Z_A +Z_B + Z_M\;,
\end{eqnarray}
 where the operators $Z_A$ and $Z_B$ count, respectively, the number of memory elements of $A$ and $B$ which are in the subspace spanned by the message
 vectors $|0\rangle$ and $|1\rangle$ (i.e the number of $a_k$ and $b_{\ell}^{(k)}$ which are NOT in $|E\rangle$), and $Z_M$ counts the number of spin up in the $M$. 
Furthermore since   at the beginning of the communication the state of the  memories of  $A$ are in   $\{|0\rangle, |1\rangle\}^{\otimes n}$, $B$  is in 
 $|E\rangle^{\otimes m}$, and $M$ is in the all spin-down state (\ref{iniziale})
  the  value of $Z$ during the whole protocol is set equal to $n$. 
 Therefore 
 the final state of the  $ABM$ system can be written as
\begin{eqnarray}
&&W ( |\psi \rangle_{A} \otimes |E\rangle_B^{\otimes m} \otimes |\omega\rangle_M
)    \label{import1} \\\nonumber  && \qquad =  \sqrt{ \eta_0} \;  |\Phi_\psi
\rangle_{AB} \otimes |\omega\rangle_M
  + \sqrt{1 -\eta_0} \; |\chi_\psi \rangle_{ABM}\;,
\end{eqnarray}
where is the probability $\eta_0$  
 to find $M$ in the ground state $|\omega\rangle_M$. By construction the vector
 $|\Phi_\psi\rangle_{AB}$ is an eigenstate of $Z_A+Z_B$ with eigenvalue $n$ (i.e. it contains exactly $n$ excitations)  while
 $|\chi_\psi\rangle_{ABM}$ satisfies  the following condition 
\begin{eqnarray} \label{ortogonalita}
{_M \langle \omega}  |\chi_\psi \rangle_{ABM} = 0\;, \quad
{_{AB} \langle \Phi_{\psi'}}  |\chi_\psi \rangle_{ABM} = 0\;,
\end{eqnarray}
for all inputs $|\psi\rangle_A$ and $|\psi'\rangle_A$
(the first inequality follows from the fact that $|\chi_\psi \rangle_{ABM}$
has at least
one excitation in $M$, the second from the fact that 
 it contains strictly less than $n$ excitations in $AB$).
The quantity $\eta_0$ may in general depends on the input $|\psi\rangle_A$. However it can be lower bounded by the joint probability that for at each step of the protocol {\em i)} Alice succeeds in inserting her
 qubit in $M_A$ and {\em ii)} Bob absorbs it in the associated sub-register, i.e. 
  \begin{eqnarray}
 \eta_0 \geqslant  \prod_{k=1}^n p_k  \;, \label{sss}
 \end{eqnarray}
with  $p_k$ being the probability that at the $k$-th step the info content of $a_k$ moves to $B^{(k)}$.  By induction  on can easily convince oneself
that the list of events associated with the probability $\Pi_n:= \prod_{k=1}^n p_k$ refer to a trajectory  in which for the whole duration of the  protocol   $M$ contains no more than
one excitation per chain. By the previous discussion of the coupling~(\ref{swap}) we know that for such states Bob's iterative procedure induces a relaxing behavior  that drives the chains toward the ground state with a probability (i.e. $p_k$ ) that increases with  $m_k$. Thus by choosing $m_k$ sufficiently big we can make
$p_k$ and  $\eta_0$ arbitrarily close to $1$ (for instance by choosing $p_k \geqslant  1 - \epsilon^{k}$ we get 
 $\eta_0\sim 1 -\epsilon$).
 The probabilities  $p_k$ possess yet another important property: {\em they do not depend upon the information $|\psi\rangle_A$ encoded into the memories $a_k$} (i.e. the value of 
 $p_k$ is independent  from the fact that $a_k$ was in $|0\rangle_{a_k}$ or in $|1\rangle_{a_k}$). This is an extremely 
 important property which comes from the fact that we are using a Dual Rail encoding: it will allow us to simplify the whole analysis and will provide us a simple
 way to construct a working transmission protocol.
Consider next the component $|\Phi_\psi \rangle_{AB}$ associated with $\eta_0$. 
By conservation of $Z$ this state can be decomposed in two orthogonal terms:
\begin{eqnarray}
|\Phi_\psi\rangle_{AB} = {\sqrt{\tfrac{\Pi_n}{\eta_0}}}
|E\rangle_A^{\otimes n}  \otimes |\Phi_\psi\rangle_B
 + { \sqrt{ 1- \tfrac{\Pi_n}{\eta_0}}}|\Delta_\psi\rangle_{AB}. \label{fq}
\end{eqnarray} 
The first component refers to the event associated with $\Pi_n$: here all excitations of the input state have moved
to $B$ (whose state is eigenvector of $Z_B$ with eigenvalue $n$) leaving $A$ into the state $|E\rangle_A^{\otimes n}$.
It is worth stressing that for the same reason why $\Pi_n$ is independent from the input state $|\psi\rangle_A$, then
$|\Phi_\psi\rangle_B$ is  a {\em faithful} encoding of such vector. 
This can be expressed by saying that  there  exists an isometry  from $A$ to $B$,
 which connects 
$|\psi\rangle_A$ to 
$|\Phi_\psi\rangle_B$, 
or equivalently that there exists a unitary transformation ${U'}$ on $AB$ which is independent from $| \psi \rangle_{A}$  such that 
 \begin{eqnarray}\label{faf}
{U'}(| \psi \rangle_{A} \otimes |E\rangle_B^{\otimes m} )= 
|E\rangle_A^{\otimes n}  \otimes |\Phi_\psi \rangle_B\;.
\end{eqnarray} 
A particular property of $|\Phi_\psi\rangle_B$ that is important to stress is the fact that {\em each one} of  the registers $B^{(k)}$ 
that compose $B$ contains exactly {\em one} excitation (specifically this is that same excitation that is initially contained in the $k$-th memory cell of Alice). 
The second component  of Eq.~(\ref{fq}) contains elements orthogonal to 
$|E\rangle_A^{\otimes n}  \otimes |\Phi_\psi\rangle_B$:
these include terms which have at least one 
of the $n$ excitations in $A$, plus terms which have no
excitation in $A$ but which have them in the "wrong" places of $B$ with respect to the components $|E\rangle_A^{\otimes n}  \otimes |\Phi_\psi\rangle_B$, i.e. they have not exactly one excitation in 
each one of the sub-registers $B^{(k)}$. 
In particular this implies that  $|\Delta_\psi\rangle_{AB}$ is orthogonal to just the $B$ component of 
$|E\rangle_A^{\otimes n}  \otimes |\Phi_{\psi'}\rangle_B$ for all possible
inputs $|\psi\rangle_A$ and $|\psi'\rangle_A$, i.e. 
\begin{eqnarray} \label{ffsf}
{_B \langle} \Phi_{\psi'}| \Delta_{\psi}\rangle_{AB}=0\;.
\end{eqnarray} 
Let us now compare Eq.~(\ref{import1}) with 
 Eq.~(\ref{imp}) that only applies for input states $|\psi\rangle_A$   belonging to the subspace { ${\cal Q}_A^{(n)}\subseteq
\{|0\rangle, |1\rangle\}^{\otimes n}$} which allows efficient communication from ${A}_{in}$ to ${A}_{out}+B_{out}$
(once we have choosen the values of  $m_k$ such subspace always exists by the analysis of Sec.~\ref{sec:flow}). 
From the orthogonality conditions of Eq.~(\ref{ortogonalita})
and from  the uniqueness of the Schmidt decomposition
 it follows that
that there must exist a component of~(\ref{imp})
 (say the one corresponding to $j=0$)  which coincides with the vector
$\sqrt{ \eta_0} \;  |\Phi_{\psi} \rangle_{AB}\otimes
|\omega\rangle_M$ (the only requirement being preventing $\eta_0$ from being
degenerate --- this can always be enforced by choosing $\Pi_n >1/2$).  
This  implies $\lambda_0 = \eta_0$, $|\zeta_0 \rangle_M = |\omega\rangle_M$, and
\begin{eqnarray}
 |\Phi_{\psi} \rangle_{AB} =
  U_0 \; ( |\psi \rangle_A \otimes |E\rangle_B^{\otimes m}) \;.
\end{eqnarray}
where we used the fact that $P_0U_0 = U_0 P$ and that { $ |\psi \rangle_A \in {\cal Q}_A^{(n)}$}.
Putting this together with Eqs.~(\ref{fq}) and (\ref{faf}) we get that also $|\Delta_\psi\rangle_{AB}$ must provide us with a unitary encoding of
the input state, i.e.  there must exists a unitary ${U''}$ 
on $AB$ such that 
\begin{eqnarray} \label{uu}
 |\Delta_\psi\rangle_{AB} = { U ''}\; ( |\psi \rangle_A \otimes |E\rangle_B^{\otimes m}) \;,
\end{eqnarray}
(by linearity this in turn implies that the
probability $\eta_0$ is constant for all $|\psi\rangle_A \in {\cal Q}_A^{(n)}$). 

Another important thing we can learn from the comparison of  Eq.~(\ref{import1}) with  Eq.~(\ref{imp}) is obtained by taking the reduced density operator of $AB$. 
 In particular Eq.~(\ref{import1}) yields the following block matrix 
  \begin{eqnarray}
\label{import11} 
{ \eta_0} \;  |\Phi_{\psi} \rangle_{AB} \langle \Phi_\psi | + (1-\eta_0)
\mbox{Tr}_M[ |\chi_\psi\rangle_{ABM}\langle \chi_\psi|] \;,\end{eqnarray}
which for input states of ${\cal Q}_A^{(n)}$  must coincide with Eq.~(\ref{dd}) which is also in block form.
Since we have already identified $ |\Phi_{\psi} \rangle_{AB} \langle \Phi_\psi |$ with the $j=0$ block of Eq.~(\ref{dd})
this implies that  the remaining term of Eq.~(\ref{import11}) must fit on the $j\geqslant 1$ blocks of Eq.~(\ref{dd}), i.e. 
  \begin{eqnarray}
\label{fff}  
&&(1-\eta_0)
\mbox{Tr}_M[ |\chi_\psi\rangle_{ABM}\langle \chi_\psi|] 
 \nonumber \\
&&= 
\sum_{j}   \lambda_j 
\; P_j U_j \; ( \rho_A \otimes |E\rangle_B\langle E|^{\otimes m}) \; U_j^\dag P_j\;.
\end{eqnarray}
Let us now take the projection of Eq.~(\ref{import11})  into the subspace orthogonal to  the one that supports the vectors $|E\rangle_A^{\otimes n}  \otimes |\Phi_\psi\rangle_B$ (this is the subspace which contains all vectors which have strictly less than $n$ excitation on $B$ plus the vectors which have $n$ excitations located in the {\em wrong places}). 
This is 
  \begin{eqnarray} \label{si}
\sigma_{AB}   \label{import111} 
&=& \tfrac{ \eta_0 - \Pi_n}{1-\Pi_n}  |\Delta_{\psi} \rangle_{AB} \langle \Delta_\psi | \nonumber \\
&+& (1-\eta_0)
\mbox{Tr}_M[ |\chi_\psi\rangle_{ABM}\langle \chi_\psi|] \;.
\end{eqnarray}
The important observation here is to notice that notwithstanding the projection the information about the input state is still well
preserved into $AB$. For the first component this is a consequence of Eq.~(\ref{uu}) while for the second one this comes  from Eq.~(\ref{fff}).

\subsection{The protocol}

Now we have all elements to construct the protocol.

\begin{itemize}
\item The first step is to fix the numbers $m_k$ of the memories which compose  Bob's $k$-th sub-register $B^{(k)}$: 
the choice  of these parameters is determined by how close to $1$ we want $\Pi_n$ of Eq.~(\ref{sss}).  This will determine the probability of success.
 For instance let us assume that we select $\Pi_n > 1 -\epsilon$ for some given $\epsilon >0$ (in any case we will require $\epsilon<1/2$ to make sure that  $\eta_0$ is a non degenerate Schmidt coefficient).  
\item Alice writes her messages into the quantum error correcting code { ${\cal Q}_A^{(n)}$} of $A$ which is associated with the choice of $m_k$ of the previous point.
\item Alice and Bob start their sequences of repetitive  operations in which the memories $A$ and $B$ are put in contact with $M$ according to Eq.~(\ref{swap}).
\item At the end of the coupling stage the state of the system is as in Eq.~(\ref{imp}). At this point Bob performs a two values projective measurement on $B$ to check wether or not
his memory $B$ fits into the subspace that
support the vectors  $|\Phi_\psi\rangle_B$ of Eq.~(\ref{fq})
(i.e. it projects $B$ into a subspace that have exactly {\em  one
excitation in each of the sub-registers} $B^{(k)}$).
This measurement is the analogous to the parity check of the Dual Rail protocol. 
\item From Eqs.~(\ref{import1}) and (\ref{fq}) 
we know that  if the result of the measurement is YES then the total system will be projected into the state 
\begin{eqnarray} 
|E\rangle_A^{\otimes n}  \otimes |\Phi_\psi\rangle_B 
 \otimes |\omega\rangle_M\;,
\end{eqnarray}
which thanks to Eq.~(\ref{faf}) contains in $B$ the input information sent by Alice (to recover it Bob needs only to perform a local operation on $B$).
The YES outcome happens with probability $\Pi_n$. 
\item With probability $1-\Pi_n<\epsilon$ the local measure performed by Bob will produce the outcome NO projecting $AB$ into the state of Eq.~(\ref{si}) which still retains the coherence of the input message.

\end{itemize}

Adopting the above procedure, with probability $\Pi_n$  we can thus transfer more than $n  - \log_2 [d_M^2 (d_M^2 +1)]$  
 qubits of information from $A_{in}$ to $B_{out}$ by using $n$ SWAP-in operation
 on Alice side. 
 Without doing anything else, on can hence achieve an average transfer fidelity which is larger than $\Pi_n$ and which can be made arbitrarily close to one by a proper choice of the numbers $m_k$ of Bob's memories.  
How good is this result compared with other techniques? 
In effect the same average transfer fidelity is attainable  
without the PM encoding strategy by only adopting a pure Dual Rail downloading technique~\cite{DB}.
 The relevance of the protocol however relies on the fact that, {\em even when Bob's
 measurement produces the NO outcome} (this happens with probability $1-\Pi_n$)
 {\em Alice messages are not completely lost}. 
 In this case in fact information still retains its coherence  but it is "delocalized"  among $A_{out}$ and $B_{out}$. 
 Alice and Bob can thus still try to get it back increasing the transmission fidelity to $1$. In the following we present
 a couple of possible strategies that can be adopted to achieve this goal. 

 The  easiest solution for compensating  the NO event
   is to admit that Alice and Bob are provided with some shared e-bit. In this case   when getting the NO outcome from his measurement
Bob can ask Alice to teleport to him all her $n$ memory elements. 
Since there are $n$ qutrits these takes $n \log_2 3$ e-bit and $2n \log_2 3$ bits of classical communication from Alice to Bob.
After teleportation Bob has direct access to the state Eq.~(\ref{si}) which still encodes perfectly the input message: to recover it he has only to perform 
a projective measurement that distinguishes the subspaces associated with the projectors $P_j$ of Eq.~(\ref{imp}) (notice that $|\Delta_\psi\rangle_{AB}$ is still 
within the $P_0$) and apply the proper unitaries transformations.
 Considering that the probability of the NO outcome is $1-\Pi_n$ 
 the average cost  of the whole procedure will thus be 
 $(1-\Pi_n) \times  \log_2 3$ e-bit and $(1-
\Pi_n)\times 2 \log_2 3$ bits of classical communication from Alice to Bob per 
transmitted qubit~\cite{NOTA11}.

An alternative solution relies in first depleting the link $M$ from all its
excitations resetting it to $|\omega_M\rangle$. This will require some sort of external (not necessarily coherent) control on $M$.  It is worth noticing that 
since the messages are safely confined in the memories $AB$
the resetting operation can be done without compromising them 
(in other communication scenarios~\cite{BOSEREV} this  will not be the case). 
Once more Alice can thus try to transfer  to Bob the content of her memory qubits by adopting the same protocol we considered before, i.e. by recursively coupling them with  $M_A$. 
 It should be notice however that in order to do so she has first to {\em translate} the information contained
in the memories $A$ into the coding subspace ${\cal Q}_A^{(n)}$. Consider that after the first step of the protocol, each one of  Alice memories $a$
span a three dimensional space generated by $|0\rangle$, $|1\rangle$ and $|E\rangle$ (compare this with the input state of the memories {\em before} the fist transmission stage when Alice messages spanned only the subspace generated by $|0\rangle$ and $|1\rangle$). In total this corresponds to $n \log_2 3$ qubits that we need to transfer. Such subspace can be fitted in the coding subspace ${\cal Q}_A^{(n_1)}$ of $n_1 = n \log_2 3 + \log_2[d_M^2 (d_M^2 +1)]$ 
new carriers (the extra term $\log_2[d_M^2 (d_M^2 +1)]$ being 
required to compensate with the fact that ${\cal Q}_A^{(n_1)}$ is 
a proper subspace of ${\cal H}_A^{(n_1)}$).
Applying the transferring protocol to these new carriers will
allow Bob to recover the messages with probability $\Pi_{n_1}$, and
will increase the average fidelity from $\Pi_{n}$ to
$\Pi_n + (1-\Pi_n) \Pi_{n_1}$. The average number of channel uses employed in the process will be instead equal to
$n + n_1(1-\Pi_n)$, corresponding to a rate $\tfrac{n}{n+n_1(1-\Pi_n)}\simeq \tfrac{1}{1 + (1 -\Pi_n) \log_2 3}$. Reiterating the whole procedure many times under the assumption of uniform success probabilities, $\Pi_{n} = \Pi_{n_1} =  \Pi_{n_2} =\cdots = 1 -\epsilon$,
will gives us unitary fidelity with asymptotic rate $\simeq 1 - \epsilon \log_2 3$ (in deriving this expression we assumed $\epsilon \log_2 3 <1$).

\subsection{Two-chains/one qubit and beyond}

The protocol we discussed in the previous section uses $L=2$ 
parallel spin chains to transfer $\sim 1$ qubit of info per each channel
use (i.e. per each couplings $S_a$), i.e.  $\sim 1/2$ qubit per channel use per
chain. Since  each chain should be capable to transport one qubit of info per swap this is not very efficient, instead we would like to have $\sim 1$
qubit per channel use per chain. 
One way to attain such rate is to replace the standard Dual Rail 
encoding~\cite{DB} with the generalized
dual rail encoding of Ref.~\cite{DGB} which, in the limit of  $L\gg1$ parallel chains, allows one to transfer $\sim 1$ one qubit per parallel chain. 
In this case the local operations~(\ref{swap}) will be replaced by analogous
coupling transformations which couple  the first (last)  spins of the $L$ chains with Alice (Bob) memory elements.

\section{Conclusion}

{ 

In conclusion, we have review the physics of PM channels by generalizing it to the case of sub-exponential
environment. We also established a connection between quantum link communication and PM, by presenting 
a general multi-use protocol which allows Alice and Bob to 
faithfully communicate through a spin chain quantum link without resetting it at each channel transmission. The protocol succeed in faithfully 
transferring the messages with an arbitrarily high
probability that can be tuned by means of Bob's operations.  The protocol originates from the merging of two apparently distinct ideas:
the codes for PM channels and the mixing property of quantum channels. This is a new approach which in principle can be exploited in other contexts.

We believe that this paper paves the way for deeper studies on quantum links in communication scenarios. One possible direction to explore is the relation between the amount of resources needed and the invoked coupling Hamiltonian.

}

\subsection*{Acknowledgments}
{
V.G. acknowledges  financial support of the Quantum Information Program of Centro Ennio De Giorgi of SNS.
D.B. acknowledges support from the QIP-IRC and Wolfson College, Oxford.
S.M. acknowledges financial support from the FET-Open grant agreement CORNER, 
number FP7-ICT-213681. 
}
 


\begin{references}
\bibitem{NOTA2} We explicitly allow for several  Bob's qubits for each single Alice qubit to  cope
 with possible dispersion effects in the information propagation through the quantum link.
\bibitem{NOTA11} Admitting the possibility of prior shared entanglement among Alice and Bob, a variant of the Dual Rail protocol~\cite{DB} can be found in which the average transmission  fidelity can be made unitary by means of 
$(1-\Pi)$ e-bits and  $2$ bits of classical communication per transmitted qubit.
This variant requires Alice NOT to transfer her messages directly but instead 
 to send to Bob one half $n$ e-bits couples that she has locally prepared. In case the  transmission succeeds Alice will then safely teleport to Bob her messages. If instead
 the  transmission fails she will use the shared entangled bit to teleport the messages.
\end{references}
\end{document}